\pdfoutput=1
\documentclass[11pt]{article}

\usepackage[T1]{fontenc}
\usepackage{amsmath,amssymb,amsthm,mathtools}
\usepackage{enumerate}
\usepackage{xcolor}
\usepackage{xurl}
\usepackage[margin=1in]{geometry}
\usepackage[hidelinks,hypertexnames=false]{hyperref}

\theoremstyle{plain}
\newtheorem{theorem}{Theorem}
\newtheorem{proposition}[theorem]{Proposition}
\newtheorem{lemma}[theorem]{Lemma}
\newtheorem{corollary}[theorem]{Corollary}
\theoremstyle{definition}
\newtheorem{remark}[theorem]{Remark}

\newcommand{\val}[2]{[#1]_{#2}}
\newcommand{\NpD}{\mathcal{N}_{\mathrm{pD}}}
\newcommand{\Dyck}{\mathcal{D}}
\newcommand{\DP}{\mathcal{D}_{\mathrm{P}}}
\newcommand{\Pset}{\mathcal P}

\title{Additive Bases from Primitive Dyck Words:\\
Regular Underapproximations, Motzkin Coding, and Digit Lifting}
\author{Takayuki Kuriyama\\
\small Growup Learning Academy, Tokyo, Japan\\
\small \texttt{growup.kuriyama@gmail.com}}
\date{}

\begin{document}
\maketitle

\begin{abstract}
We study additive representations by integers whose canonical binary
expansions are primitive Dyck words.  Pairing consecutive bits yields a
positional form of the classical relation between Dyck paths and
two-coloured Motzkin paths: except for $10$, primitive Dyck words are
exactly the binary block images of base-$4$ words $3w0$, where $w$ is a
two-coloured Motzkin word.  This exposes a regular underapproximation,
digit closure, and sharp generation bounds.  We prove an interval
digit-lifting theorem for digitally closed sets and a constructive
base-$4$ propagation algorithm that lifts finite sumset certificates to
infinite tails in logarithmically many recursive stages.  Combining these
tools with exact finite certificates and generation-gap lower bounds, we
classify all positive even integers requiring more than six primitive Dyck
summands.  The integer $46$ requires eight, and
$34,44,98,154,198,202,206,838,842,846$ require seven; every other positive
even integer requires at most six.  Thus $848$ is the sharp eventual
threshold.  The bound is asymptotically optimal because
$10\cdot4^{k+1}-6$ requires six summands for every $k\geq2$.  The associated
halved family has exact asymptotic additive order five.  Supplementary
programs reproduce all finite certificates using exact integer arithmetic.
\end{abstract}



\section{Introduction}\label{sec:introduction}

Numeration systems provide a natural interface between formal-language
theory and additive number theory.  Given a language $L$ over a digit
alphabet, one may read the words of $L$ as integers and ask whether the
resulting value set is an additive basis, what its least order is, and
whether the exceptional targets can be classified effectively.  For
regular languages, these questions can often be approached through finite
automata.  Bell, Hare, and Shallit gave general results for automatic
additive bases~\cite{BHS}, and related methods have produced sharp results
for binary palindromes~\cite{RSS} and binary squares~\cite{MNRS}.

Bell, Lidbetter, and Shallit applied regular-language approximation to the
language of balanced binary words and obtained strong additive theorems
for the corresponding totally balanced numbers~\cite{BLS}.  Their work is
a particularly close predecessor of the present paper.  Regular
approximations of context-free languages have a broader history, including
constructions of both underapproximations and overapproximations by finite
automata~\cite{Nederhof,MohriNederhof}.

Language-based numeration has also been studied in the framework of
abstract numeration systems~\cite{RigoNumeration}.  In particular,
Charlier, Le Gonidec, and Rigo considered a generalized numeration system
built on the Dyck language~\cite{CLGR}.  Their numerical map is induced by
an ordering of the underlying language, whereas the present paper reads
each word by its ordinary positional binary value.  On the combinatorial
side, bijections relating Dyck paths or ordered trees to two-coloured
Motzkin paths are classical~\cite{DeutschShapiro}.  Our use of Motzkin
paths is tailored to the paired binary digits and hence retains exact
base-$2$/base-$4$ positional values.

The proof below is hybrid in a precise but deliberately limited sense.  An
explicit regular underapproximation supplies the six-summand covering used
for targets divisible by $4$.  For targets congruent to $2$ modulo $4$, our
argument instead uses a five-summand covering by the full Motzkin-coded
family.  We do not claim that every possible regular underapproximation
must fail for the latter residue class.

Let $\Dyck$ be the Dyck language over $\{1,0\}$, with $1$ interpreted as
an up-step and $0$ as a down-step.  A nonempty Dyck word is
\emph{primitive} if its path returns to height zero only at the endpoint.
Let $\DP$ denote the language of primitive Dyck words.  Since
$\DP=1\Dyck0$, it is deterministic context-free.  It is not regular,
because
\[
 \DP\cap1^*0^*=\{1^n0^n:n\geq1\}.
\]
We read the words of $\DP$ as binary integers and adjoin $0$; the resulting
set is denoted by $\NpD$.

The primitive restriction is substantial.  Every nonempty Dyck word
factors uniquely into primitive factors, but an additive representation in
our problem requires each summand itself to be represented by a single
primitive word.  Thus a representation by unrestricted balanced words does
not preserve the number of summands after passage to primitive factors.
Moreover, the language $\DP$ is not available to finite-automaton methods
in its entirety.

The key structural step is a positional block coding.  Pairing consecutive
binary digits produces base-$4$ digits, and the primitive condition becomes
a Motzkin prefix condition.  We prove the exact language identity
\[
 \DP=\{10\}\cup\phi(3\mathsf{Motz}0),
\]
where $\phi(0)=00$, $\phi(1)=01$, $\phi(2)=10$, and $\phi(3)=11$.
The novelty used here is not the general existence of Dyck--2-Motzkin
bijections, but the compatibility of this primitive-word block map with
canonical positional value.  The regular language $11(01\mid10)^*00$
then supplies one additive covering, while the full Motzkin-coded family
supplies the digit closure used in the other covering.

Our contributions are fourfold.
\begin{enumerate}[(i)]
\item We adapt the classical Dyck--2-Motzkin viewpoint to canonical binary
valuation and obtain an exact base-$4$ characterization of the primitive
Dyck value set, together with sharp generation bounds.
\item We separate two complementary roles in the additive proof: an
explicit regular underapproximation covers the multiples of $4$, while a
five-summand theorem for the full Motzkin-coded family covers the auxiliary
targets needed in the residue class $2$ modulo $4$.
\item We formulate a general interval digit-lifting theorem, prove a
base-$4$ tail-propagation principle, and extract a constructive algorithm
that converts a finite table of seed representations into representations
throughout the infinite tail.
\item We determine the complete exceptional set, the sharp eventual
threshold, and the exact relative asymptotic order; generation gaps also
give an infinite lower-bound sequence.
\end{enumerate}

For a positive even integer $N$, let $r(N)$ be the least number of positive
primitive Dyck numbers whose sum is $N$.  This is well defined because
$10\in\DP$ represents $2$, so every positive even $N$ is a sum of $N/2$
copies of $2$.  Our first main result classifies exactly the values above
level six.

\begin{theorem}\label{thm:classification}
The values of $r(N)$ exceeding six are exactly as follows:
\[
\begin{aligned}
 r(46)&=8,\\[4pt]
 r(N)=7\quad&\Longleftrightarrow\quad
 N\in\{34,44,98,154,198,202,206,838,842,846\}.
\end{aligned}
\]
Every other positive even integer satisfies $r(N)\leq6$.  Equivalently,
$848$ is the least even integer $T$ such that every even $N\geq T$ is a
sum of at most six primitive Dyck numbers.
\end{theorem}

The eventual bound cannot be lowered.

\begin{theorem}\label{thm:asymptotic-order}
Let
\[
 h_{\mathrm{asym}}
 =\min\{h:\text{every sufficiently large even }N
              \text{ satisfies }r(N)\leq h\}.
\]
Then $h_{\mathrm{asym}}=6$.  More precisely, for every integer $k\geq2$,
\[
 N_k=10\cdot4^{k+1}-6
\]
satisfies $r(N_k)=6$.
\end{theorem}

Section~\ref{sec:preliminaries} fixes the language and additive notation.
Section~\ref{sec:base4} proves the positional Motzkin block coding and
generation bounds.  Section~\ref{sec:lifting} develops digit lifting, its
constructive tail algorithm, the regular underapproximation, and the finite
certificates.  Section~\ref{sec:exceptions}
proves the complete classification, and Section~\ref{sec:asymptotic}
proves asymptotic optimality.  The supplementary verification algorithms
are described in Section~\ref{sec:verification}.

\section{Languages, values, and additive notation}\label{sec:preliminaries}

A word $w\in\{1,0\}^*$ is a \emph{Dyck word} if it contains equally many
$1$'s and $0$'s and every prefix contains at least as many $1$'s as
$0$'s.  The empty word is denoted by $\lambda$.  A nonempty Dyck word is
\emph{primitive} if no proper nonempty prefix is balanced.  Equivalently,
its height path first returns to zero at the endpoint.  We write $\Dyck$
for the Dyck language and $\DP$ for its primitive sublanguage.  Every
primitive word is uniquely of the form $1u0$ with $u\in\Dyck$.

For a word $w$ over the base-$b$ digits, let $\val{w}{b}$ denote its
numerical value.  For a language $L$ over those digits, write
\[
 \operatorname{val}_b(L)=\{\val{w}{b}:w\in L\}.
\]
Define
\[
 \NpD=\{0\}\cup\{\val{w}{2}:w\in\DP\}
 \subseteq\mathbb Z_{\geq0}.
\]
Every positive member of $\NpD$ is even.  The following stronger terminal
property will be used repeatedly.

\begin{lemma}\label{lem:mod4}
The number $2$ is the only positive primitive Dyck number congruent to
$2$ modulo $4$.  Every other positive primitive Dyck number is divisible
by $4$.
\end{lemma}

\begin{proof}
Let $w\in\DP$.  Its last letter is $0$.  If $|w|=2$, then $w=10$ and
$[w]_2=2$.  Suppose $|w|\geq4$ and write $w=ub0$ with $b\in\{0,1\}$.
If $b=1$, the suffix $10$ has net height zero, so the proper prefix $u$
is balanced, contradicting primitivity.  Hence $b=0$, the word ends in
$00$, and $4\mid[w]_2$.
\end{proof}

Among the $C_n$ Dyck words of length $2n$, exactly $C_{n-1}$ are
primitive, so their proportion tends to $1/4$.  Thus primitive words are
not sparse at a fixed length; the additive obstructions below arise from
their hierarchical value structure rather than from a vanishing language
density.

For a nonnegative integer $k$ and a set $X$ of nonnegative integers, let
$kX$ denote the set of sums of exactly $k$ not necessarily distinct
members of $X$, with $0X=\{0\}$.  Put
\[
 F_k(X)=\bigcup_{j=0}^{k}jX.
\]
Thus $F_k(X)$ is the set of sums of at most $k$ elements of $X$.  If
$0\in X$, membership in $kX$ may be read as a representation by at most
$k$ positive elements of $X$: the zero entries are only padding and are
deleted when the representation is transferred back to positive primitive
Dyck summands.  This is not dilation notation: for a scalar $c$, we write
$c\cdot X=\{cx:x\in X\}$.

Lemma~\ref{lem:mod4} motivates the halved positive family
\begin{equation}\label{eq:P-def}
 \Pset=\{v/4:v\in\NpD,\ v\geq12\}.
\end{equation}
Then
\begin{equation}\label{eq:NpD-decomp}
 \NpD=\{0,2\}\cup4\cdot\Pset.
\end{equation}
The first values in $\Pset$ are
\[
 3,13,14,53,54,57,58,60,213,214,217,\ldots.
\]

\section{Two-coloured Motzkin block coding}\label{sec:base4}

Two-coloured Motzkin paths are classically related bijectively to Dyck
paths and ordered trees; see, for example, Deutsch and
Shapiro~\cite{DeutschShapiro}.  The construction below is a direct
left-to-right block map for primitive Dyck words.  Its role in this paper
is value-sensitive: the same pairs are simultaneously the base-$4$ digits
of the represented integer.

Binary words of even length may be read blockwise in base $4$ through
\[
 00\leftrightarrow0,\qquad
 01\leftrightarrow1,\qquad
 10\leftrightarrow2,\qquad
 11\leftrightarrow3.
\]
Define the uniform morphism
\[
 \phi:\{0,1,2,3\}^*\longrightarrow\{0,1\}^*
\]
by
\[
 \phi(0)=00,\quad \phi(1)=01,\quad
 \phi(2)=10,\quad \phi(3)=11.
\]
Assign weights
\[
 \sigma(3)=1,\qquad \sigma(0)=-1,\qquad
 \sigma(1)=\sigma(2)=0,
\]
and extend $\sigma$ additively to words.  A word over
$\{0,1,2,3\}$ is a \emph{two-coloured Motzkin word} if every prefix has
nonnegative weight and its total weight is zero.  Let $\mathsf{Motz}$
denote this language; it contains $\lambda$.

\begin{theorem}[Block-coding characterization]\label{thm:block-coding}
The language of primitive Dyck words satisfies the disjoint identity
\begin{equation}\label{eq:block-coding-language}
 \DP=\{10\}\cup\phi\bigl(3\mathsf{Motz}0\bigr).
\end{equation}
\end{theorem}

\begin{proof}
The word $10$ is primitive.  Let $B\in\DP$ have length at least $4$, and
group its bits into pairs.  Write
\[
 B=\phi(q_1q_2\cdots q_L),
 \qquad q_i\in\{0,1,2,3\}.
\]
Let $h_B(t)$ be the binary height after $t$ bits.  Across a complete block,
\begin{equation}\label{eq:block-height}
 h_B(2j)-h_B(2j-2)=2\sigma(q_j).
\end{equation}
Primitivity forces the first block to be $11$ and the last block to be
$00$; otherwise the path would return to height zero at position $2$ or
at position $2L-2$.  Hence $q_1=3$ and $q_L=0$.

Put $u=q_2\cdots q_{L-1}$.  At each internal block boundary,
$h_B(2j)$ is a positive even integer.  Summing
\eqref{eq:block-height} therefore gives
\[
 \sum_{i=1}^{j}\sigma(q_i)\geq1
 \qquad(1\leq j\leq L-1).
\]
Subtracting $\sigma(q_1)=1$ shows that every prefix of $u$ has
nonnegative weight.  Since the total block weight is zero and the endpoint
weights are $1$ and $-1$, the total weight of $u$ is zero.  Thus
$u\in\mathsf{Motz}$ and $B=\phi(3u0)$.

Conversely, let $u=q_2\cdots q_m\in\mathsf{Motz}$ and set
$B=\phi(3u0)$.  At each nonterminal even block boundary the binary height
is
\[
 2\left(1+\sum_{i=2}^{j}\sigma(q_i)\right)\geq2,
\]
and the total height is zero.  It remains to inspect odd positions.  A
block $3$ or $2$ begins with $1$ and therefore increases the height.  A
block $1$ begins with $0$, but it starts at an even boundary of height at
least $2$, so the intermediate height is at least $1$.  If a block $0$
occurs inside $u$, then the Motzkin weight after reading that block must
remain nonnegative.  Since the block has weight $-1$, its preceding weight
is therefore at least $1$; hence the block starts at height at least $4$
and its first $0$ leaves positive height.  Finally, the appended terminal block
$0$ starts at height $2$, passes through height $1$, and ends at $0$.
Thus every proper prefix of $B$ has positive height, so $B\in\DP$.
\end{proof}

For a positive integer $n$, let $\operatorname{quad}(n)$ denote its
standard base-$4$ expansion.

\begin{corollary}[Base-$4$ value characterization]\label{cor:base4}
A positive integer $p$ lies in $\Pset$ if and only if there exists
$w\in\mathsf{Motz}$ such that
\begin{equation}\label{eq:base4-value-characterization}
 \operatorname{quad}(p)=3w.
\end{equation}
\end{corollary}

\begin{proof}
By definition, $p\in\Pset$ exactly when $4p$ is represented by a primitive
Dyck word of length at least $4$.  Multiplication by $4$ appends the
base-$4$ digit $0$.  Theorem~\ref{thm:block-coding} therefore says exactly
that the base-$4$ expansion of $4p$ is $3w0$ with
$w\in\mathsf{Motz}$, equivalently that
$\operatorname{quad}(p)=3w$.
\end{proof}
\begin{corollary}\label{cor:Pclosure}
If $p\in\Pset$, then $4p+1$ and $4p+2$ also belong to $\Pset$.
\end{corollary}

\begin{proof}
By Corollary~\ref{cor:base4}, there exists $w\in\mathsf{Motz}$ such that
\[
 \operatorname{quad}(p)=3w.
\]
Then
\[
 \operatorname{quad}(4p+1)=3w1,
 \qquad
 \operatorname{quad}(4p+2)=3w2.
\]
Since
\[
 \sigma(1)=\sigma(2)=0,
\]
appending either digit preserves membership in $\mathsf{Motz}$.  The claim follows from
Corollary~\ref{cor:base4}.
\end{proof}

We next determine the smallest and largest values in each generation.

\begin{lemma}[Extremal two-coloured Motzkin values]\label{lem:extremal}
For every $\ell\geq0$ and every two-coloured Motzkin word $w$ of length $\ell$,
\begin{equation}\label{eq:motzkin-extrema}
  \frac{4^{\ell}-1}{3}\leq[w]_4\leq4^{\ell}-2^{\ell}.
\end{equation}
Both bounds are attained.  The lower bound is attained by $1^{\ell}$, and the
upper bound is attained by
\[
  w_{\max,\ell}=
  \begin{cases}
    3^a0^a,&\ell=2a,\\
    3^a2\,0^a,&\ell=2a+1.
  \end{cases}
\]
Here $1^0$ is understood to be the empty word $\lambda$.
\end{lemma}

\begin{proof}
We use induction on $\ell$.  The assertion is immediate for $\ell=0$.
Write
\[
 w=c_0c_1\cdots c_{\ell-1},
 \qquad
 c_j\in\{0,1,2,3\}.
\]

\emph{Lower bound.}
Since $\sigma(0)=-1$, the first digit cannot be $0$.  If $c_0\in\{1,2\}$, then $\sigma(c_0)=0$, so deleting the first
digit does not change any subsequent prefix weight or the total weight.
Hence the suffix $w'=c_1\cdots c_{\ell-1}$ is again a two-coloured
Motzkin word.  By the induction hypothesis for words of length $\ell-1$,
\[
  [w]_4\geq4^{\ell-1}+[w']_4
  \geq4^{\ell-1}+\frac{4^{\ell-1}-1}{3}
  =\frac{4^{\ell}-1}{3}.
\]
If $c_0=3$, then
\[
  [w]_4
  \geq[3\underbrace{00\cdots0}_{\ell-1\text{ digits}}]_4
  =3\cdot4^{\ell-1}
  \geq\frac{4^{\ell}-1}{3}.
\]
Equality is attained by $w=1^{\ell}$.

\emph{Upper bound.}
If $c_0\in\{1,2\}$, then, as above, $w'$ is again a two-coloured
Motzkin word.  By the induction hypothesis for words of length $\ell-1$,
\begin{align*}
  [w]_4
  &=[c_0c_1\cdots c_{\ell-1}]_4\\
  &\leq[2c_1\cdots c_{\ell-1}]_4\\
  &\leq2\cdot4^{\ell-1}+\bigl(4^{\ell-1}-2^{\ell-1}\bigr)\\
  &=3\cdot4^{\ell-1}-2^{\ell-1}\\
  &\leq4^{\ell}-2^{\ell}.
\end{align*}

Now suppose $c_0=3$, and let $r$ be the length of the shortest nonempty
prefix of $w$ having total weight $0$.  This prefix has the form $3u0$,
where
\[
 u=c_1\cdots c_{r-2},
\]
and let
\[
 v=c_r\cdots c_{\ell-1}
\]
be the remaining suffix.  Then
\[
 2\leq r\leq\ell,
 \qquad
 |u|=r-2,
 \qquad
 |v|=\ell-r,
\]
and $u,v\in\mathsf{Motz}$.  Thus $w=3u0v$.  By place value and the induction hypothesis,
\begin{align*}
  [w]_4
  &=3\cdot4^{\ell-1}+[u]_4\,4^{\ell-r+1}+[v]_4\\
  &\leq3\cdot4^{\ell-1}
    +\bigl(4^{r-2}-2^{r-2}\bigr)4^{\ell-r+1}
    +4^{\ell-r}-2^{\ell-r}\\
  &=4^{\ell}-2^{2\ell-r}+2^{2\ell-2r}-2^{\ell-r}\\
  &=4^{\ell}-\bigl(2^{2\ell-r}-2^{2\ell-2r}\bigr)-2^{\ell-r}.
\end{align*}
It remains to verify
\[
  2^{2\ell-r}-2^{2\ell-2r}
  \geq2^{\ell}-2^{\ell-r}.
\]
After factoring, this is
\[
  2^{2\ell-2r}(2^r-1)\geq2^{\ell-r}(2^r-1),
\]
which follows from $r\leq\ell$.  This proves the upper bound.

Finally, the displayed words $w_{\max,\ell}$ are two-coloured Motzkin words.
A geometric-series calculation gives
\[
  [w_{\max,\ell}]_4=4^{\ell}-2^{\ell},
\]
so the upper bound is attained.
\end{proof}

An element of $\Pset$ belongs to \emph{generation $d$} if its base-$4$
expansion has exactly $d$ digits.  By Corollary~\ref{cor:base4}, a
generation-$d$ element $p$ satisfies
\[
 \operatorname{quad}(p)=3w
\]
for some $w\in\mathsf{Motz}$ with $|w|=d-1$.  Therefore
\begin{equation}\label{eq:generation-value}
 p=[3w]_4=3\cdot4^{d-1}+[w]_4.
\end{equation}
The following corollary is now immediate from Lemma~\ref{lem:extremal}.

\begin{corollary}[Generation bounds]\label{cor:gen}
For every $d\geq1$, every generation-$d$ element of $\Pset$ lies in
\begin{equation}\label{eq:generation-interval}
  [g_d,U_d]
  =\left[
  3\cdot4^{d-1}+\frac{4^{d-1}-1}{3},
  4^d-2^{d-1}
  \right],
\end{equation}
and both endpoints are attained.  In particular,
\begin{enumerate}[(a)]
\item there is no element of $\Pset$ in $(U_d,3\cdot4^d)$;
\item for every $k\geq0$,
\begin{equation}\label{eq:three-g}
  3g_{k+1}=10\cdot4^k-1.
\end{equation}
\end{enumerate}
\end{corollary}

\begin{proof}
Let $p$ be a generation-$d$ element.  Applying
\eqref{eq:motzkin-extrema} with $\ell=d-1$ gives
\[
 \frac{4^{d-1}-1}{3}
 \leq[w]_4
 \leq4^{d-1}-2^{d-1}.
\]
Substituting these bounds into \eqref{eq:generation-value} yields
\eqref{eq:generation-interval}.  The lower
endpoint is attained by $w=1^{d-1}$, and the upper endpoint is attained by
$w=w_{\max,d-1}$.  Part~(a) follows because the next generation begins at
$g_{d+1}>3\cdot4^d$.  Finally,
\[
  3g_{k+1}
  =3\left(3\cdot4^k+\frac{4^k-1}{3}\right)
  =10\cdot4^k-1.
\]
\end{proof}

\section{Regular underapproximation, digit lifting, and finite certificates}
\label{sec:lifting}

We first isolate a digit-lifting mechanism that is independent of Dyck
words.

\begin{theorem}[Interval digit lifting]\label{thm:general-digit-lifting}
Let $b\geq2$, let $\alpha\leq\beta$ be integers, and let
$A\subseteq\mathbb Z$ satisfy
\begin{equation}\label{eq:general-digit-closure}
 b\cdot A+[\alpha,\beta]\subseteq A.
\end{equation}
For every $k\geq1$ and every $m\in kA$,
\begin{equation}\label{eq:general-digit-lifting}
 [bm+k\alpha,bm+k\beta]\subseteq kA.
\end{equation}
\end{theorem}

\begin{proof}
Write $m=a_1+\cdots+a_k$ with $a_i\in A$.  Every integer
$r\in[k\alpha,k\beta]$ is a sum
$r=\eta_1+\cdots+\eta_k$ with
$\eta_i\in[\alpha,\beta]$: start with all $\eta_i=\alpha$ and distribute
$r-k\alpha$ unit increments among the $k$ coordinates, never exceeding
$\beta$.  By \eqref{eq:general-digit-closure},
$ba_i+\eta_i\in A$, and hence
\[
 bm+r=\sum_{i=1}^{k}(ba_i+\eta_i)\in kA.
\]
\end{proof}

\begin{corollary}[Base-$4$ digit lifting]\label{lem:digit-lifting}
Let $k\geq1$, and let $A\subseteq\mathbb Z_{>0}$ satisfy
\begin{equation}\label{eq:digit-closure}
 4\cdot A+\{1,2\}\subseteq A.
\end{equation}
If $m\in kA$, then
\begin{equation}\label{eq:digit-lifting}
 [4m+k,4m+2k]\subseteq kA.
\end{equation}
\end{corollary}

\begin{proof}
Apply Theorem~\ref{thm:general-digit-lifting} with
$b=4$, $\alpha=1$, and $\beta=2$.
\end{proof}

\begin{corollary}[Propagation from a finite interval to a tail]
\label{cor:tail-propagation}
Let $k\geq3$, and let $A\subseteq\mathbb Z_{>0}$ satisfy
\eqref{eq:digit-closure}.  Suppose that, for some integer $M\geq1$,
\begin{equation}\label{eq:tail-seed}
 [M,4M+k-1]\subseteq kA.
\end{equation}
Then
\[
 [M,\infty)\subseteq kA.
\]
\end{corollary}

\begin{proof}
We use strong induction on $n\geq M$.  The seed interval handles
$M\leq n\leq4M+k-1$.  Let $n\geq4M+k$, and choose the unique
\[
 \rho\in\{k,k+1,k+2,k+3\}
\]
with $\rho\equiv n\pmod4$.  Since $k\geq3$,
$k\leq\rho\leq k+3\leq2k$.  Put $m=(n-\rho)/4$.  Then $m$ is an
integer and
\[
 m\geq\frac{4M+k-(k+3)}4=M-\frac34,
\]
so $m\geq M$.  Also $m<n$.  The induction hypothesis gives $m\in kA$,
and Corollary~\ref{lem:digit-lifting} gives
\[
 n=4m+\rho\in[4m+k,4m+2k]\subseteq kA.
\]
\end{proof}

\begin{corollary}[Constructive tail representations]
\label{cor:constructive-tail}
Under the hypotheses of Corollary~\ref{cor:tail-propagation}, suppose that
one stores a representation
\[
 n=a_1+\cdots+a_k,\qquad a_i\in A,
\]
for every $n\in[M,4M+k-1]$.  Then a $k$-term representation of every
$n\geq M$ can be constructed recursively.  The construction uses
$O(\log(n/M))$ recursive stages and $O(k\log(n/M))$ applications of the
digit maps $a\mapsto4a+1$ and $a\mapsto4a+2$.
\end{corollary}

\begin{proof}
For $n$ in the seed interval, return the stored representation.  Otherwise
choose $\rho$ and $m=(n-\rho)/4$ as in the proof of
Corollary~\ref{cor:tail-propagation}, and recursively write
$m=a_1+\cdots+a_k$.  Since $k\leq\rho\leq2k$, choose digits
$\eta_i\in\{1,2\}$ whose sum is $\rho$; explicitly, take
$\rho-k$ of the digits equal to $2$ and the remaining digits equal to $1$.
Then
\[
 n=\sum_{i=1}^k(4a_i+\eta_i),
\]
and every transformed summand belongs to $A$ by
\eqref{eq:digit-closure}.  At each nonseed stage the recursive argument is
strictly less than one quarter of the current target.  Hence the recursion
reaches the seed interval after at most
$1+\lceil\log_4(n/M)\rceil$ stages, and each stage performs $k$ digit
updates.
\end{proof}
\subsection{A regular underapproximation for multiples of four}

Consider the regular language
\[
 \mathcal R=11(01\mid10)^*00.
\]
By Theorem~\ref{thm:block-coding}, one has
$\mathcal R\subseteq\DP$: after block decoding, its words are exactly
$\phi(3\{1,2\}^*0)$, and every word over $\{1,2\}$ is a two-coloured
Motzkin word.  Thus $\mathcal R$ is an explicit regular
underapproximation of the nonregular language $\DP$.

Define the corresponding regular base-$4$ value family
\[
 \mathcal E=\{0\}\cup
 \bigl\{[3\varepsilon_1\cdots\varepsilon_j]_4:
        j\geq0,\ \varepsilon_i\in\{1,2\}\bigr\},
\]
and put $\mathcal E^+=\mathcal E\setminus\{0\}$.  When $j=0$, the
suffix is the empty word, and the corresponding element is $[3]_4=3$.
Every word $\varepsilon_1\cdots\varepsilon_j$ belongs to
$\mathsf{Motz}$.  Hence Corollary~\ref{cor:base4} gives
\[
 \mathcal E^+\subseteq\Pset.
\]
More precisely, $\operatorname{val}_2(\mathcal R)=4\cdot\mathcal E^+$.  It follows from \eqref{eq:NpD-decomp} that
\begin{equation}\label{eq:4E}
 4\cdot\mathcal E\subseteq\NpD.
\end{equation}

If $e\in\mathcal E^+$, then its base-$4$ expansion has the form
\[
 e=[3\varepsilon_1\cdots\varepsilon_j]_4,
 \qquad \varepsilon_i\in\{1,2\}.
\]
Consequently,
\[
 4e+1=[3\varepsilon_1\cdots\varepsilon_j1]_4,
 \qquad
 4e+2=[3\varepsilon_1\cdots\varepsilon_j2]_4.
\]
Both integers again belong to $\mathcal E^+$.  Therefore
\begin{equation}\label{eq:Eclosure}
 4\cdot\mathcal E^++\{1,2\}\subseteq\mathcal E^+.
\end{equation}

We use only the following finite inclusions as computational certificates.

\begin{lemma}[Finite certificate for the regular family]
\label{lem:Ecertificate}
Put
\[
\begin{aligned}
 A_0&=\{0,3,13,14\},
 &B&=\{53,54,57,58\},\\
 A&=\{3,13,14,53,54,57,58\},
 &C&=\{213,214,217,218,229,230,233,234\}.
\end{aligned}
\]
Then $A_0\subseteq\mathcal E$ and $A,B,C\subseteq\mathcal E^+$.
Moreover,
\begin{equation}\label{eq:E-small-certificate}
\begin{array}{c|c}
\text{sumset}&\text{contained interval}\\ \hline
6A_0 &[25,62]\\
B+5A_0 &[56,128]\\
2B+4A_0 &[106,172]\\
3B+3A_0 &[159,216]\\
4B+2A_0 &[212,260]
\end{array}
\end{equation}
and
\begin{equation}\label{eq:E-tail-certificate}
\begin{array}{c|c}
\text{sumset}&\text{contained interval}\\ \hline
6A &[218,260]\\
C+5A &[258,524]\\
2C+4A &[446,700]\\
3C+3A &[648,876]\\
4C+2A &[858,1052]
\end{array}
\end{equation}
hold.
\end{lemma}

\begin{proof}
Each assertion is an exhaustive finite computation of a sumset.  For a
finite set $X$, start with $S_0(X)=\{0\}$ and compute recursively
\[
 S_{i+1}(X)=S_i(X)+X.
\]
The supplementary verification script checks that every integer in each
displayed interval belongs to the corresponding sumset.
\end{proof}

\begin{proposition}\label{prop:E}
Every integer $n\geq25$ belongs to $6\mathcal E$.
\end{proposition}

\begin{proof}
The five intervals in \eqref{eq:E-small-certificate} overlap and give
\[
 [25,217]\subseteq6\mathcal E.
\]
Every sumset occurring in \eqref{eq:E-tail-certificate} is formed only
from the sets $A,C\subseteq\mathcal E^+$; no zero summand is used.
Hence the five intervals in that certificate overlap to give
\[
 [218,877]\subseteq6\mathcal E^+.
\]
Combining the two finite ranges, we obtain
\begin{equation}\label{eq:E-small-cover}
 [25,877]\subseteq6\mathcal E.
\end{equation}

We apply Corollary~\ref{cor:tail-propagation} to
$A=\mathcal E^+$, rather than to $A=\mathcal E$.  The set
$\mathcal E$ contains $0$ and therefore does not satisfy
\eqref{eq:digit-closure}, since that condition would force
$1,2\in\mathcal E$.  By contrast, $\mathcal E^+$ satisfies
\eqref{eq:digit-closure} by \eqref{eq:Eclosure}.  Since
\[
 877=4\cdot218+6-1,
\]
the inclusion
\[
 [218,877]\subseteq6\mathcal E^+
\]
is exactly the seed condition with $k=6$ and $M=218$.
Corollary~\ref{cor:tail-propagation} therefore gives
\[
 [218,\infty)\subseteq6\mathcal E^+\subseteq6\mathcal E.
\]
Combining this with
\[
 [25,217]\subseteq6\mathcal E
\]
proves
\[
 [25,\infty)\subseteq6\mathcal E.
\]
\end{proof}

It follows from Proposition~\ref{prop:E} and \eqref{eq:4E} that every
multiple of $4$ at least $100$ is a sum of at most six primitive Dyck
numbers.

\subsection{Five-summand covering by the full context-free family}

The regular family $\mathcal E$ supplies the covering used for integers
congruent to $0$ modulo $4$.  For integers congruent to $2$ modulo $4$,
our proof instead uses five-summand representations from the full halved
primitive family $\Pset$.  This establishes the complementary role of the
full Motzkin-coded family in the present argument; it is not an
impossibility statement about all conceivable regular subfamilies.

We now verify explicitly the small elements of $\Pset$ used in the finite
interval certificate.  For
\[
 w=c_1c_2\cdots c_\ell\in\{0,1,2,3\}^{\ell},
\]
let
\[
 \mathbf h(w):=
 \bigl(
 \sigma(c_1),\,
 \sigma(c_1c_2),\,
 \ldots,\,
 \sigma(c_1c_2\cdots c_\ell)
 \bigr)
\]
be the sequence of weights of its nonempty prefixes.  Recall that
\[
 \sigma(3)=1,\qquad
 \sigma(0)=-1,\qquad
 \sigma(1)=\sigma(2)=0.
\]
For the empty word, put $\mathbf h(\lambda)=()$.  In this notation,
$w\in\mathsf{Motz}$ if and only if every component of $\mathbf h(w)$ is
nonnegative and its last component is $0$.

First, we enumerate all two-coloured Motzkin words of lengths $0$, $1$,
and $2$.  The corresponding elements of $\Pset$ are as follows:
\[
\begin{array}{c|c|c|c}
 |w|&w&[3w]_4&\mathbf h(w)\\ \hline
0&\lambda &[3]_4=3&()\\ \hline
1&1 &[31]_4=13&(0)\\
1&2 &[32]_4=14&(0)\\ \hline
2&11 &[311]_4=53&(0,0)\\
2&12 &[312]_4=54&(0,0)\\
2&21 &[321]_4=57&(0,0)\\
2&22 &[322]_4=58&(0,0)\\
2&30 &[330]_4=60&(1,0)
\end{array}
\]

This table is complete.  A word of length $1$ has total weight $0$ only
when its unique letter is $1$ or $2$.  For a word of length $2$, if the
first letter is $1$ or $2$, then the second letter must again be $1$ or
$2$, giving
\[
 11,\ 12,\ 21,\ 22.
\]
If the first letter is $3$, then the second letter must be $0$ in order
to return from height $1$ to total weight $0$, giving the word $30$.
A word beginning with $0$ is impossible because its first prefix has
negative weight.  Thus there are exactly
\[
 1+2+5=8
\]
two-coloured Motzkin words of lengths $0$, $1$, and $2$.

By Corollary~\ref{cor:base4}, the corresponding elements of $\Pset$ are
exactly
\begin{equation}\label{eq:L-def}
 L=\{3,13,14,53,54,57,58,60\}
 \subseteq\Pset.
\end{equation}

Next, we enumerate all two-coloured Motzkin words of length $3$:
\[
\begin{array}{c|c|c}
 w&[3w]_4&\mathbf h(w)\\ \hline
111 &[3111]_4=213&(0,0,0)\\
112 &[3112]_4=214&(0,0,0)\\
121 &[3121]_4=217&(0,0,0)\\
122 &[3122]_4=218&(0,0,0)\\
130 &[3130]_4=220&(0,1,0)\\ \hline
211 &[3211]_4=229&(0,0,0)\\
212 &[3212]_4=230&(0,0,0)\\
221 &[3221]_4=233&(0,0,0)\\
222 &[3222]_4=234&(0,0,0)\\
230 &[3230]_4=236&(0,1,0)\\ \hline
301 &[3301]_4=241&(1,0,0)\\
302 &[3302]_4=242&(1,0,0)\\
310 &[3310]_4=244&(1,1,0)\\
320 &[3320]_4=248&(1,1,0)
\end{array}
\]

This enumeration is also complete.  If the first letter is $1$, the
remaining suffix of length $2$ must be one of
\[
 11,\ 12,\ 21,\ 22,\ 30,
\]
which gives
\[
 111,\ 112,\ 121,\ 122,\ 130.
\]
The same argument for the first letter $2$ gives
\[
 211,\ 212,\ 221,\ 222,\ 230.
\]
If the first letter is $3$, the first prefix has weight $1$, so the
remaining two letters must have total weight $-1$.  The possibilities
compatible with nonnegative prefix weights are
\[
 01,\ 02,\ 10,\ 20,
\]
giving
\[
 301,\ 302,\ 310,\ 320.
\]
A word beginning with $0$ is impossible.  Hence the table contains all
\[
 5+5+4=14
\]
two-coloured Motzkin words of length $3$.

Thus the corresponding elements of generation $4$ of $\Pset$ are exactly
\begin{equation}\label{eq:H-def}
\begin{split}
 H=\{&213,214,217,218,220,229,230,233,234,236,\\
     &241,242,244,248\}
 \subseteq\Pset.
\end{split}
\end{equation}
In fact, $H$ is not merely a subset of $\Pset$: it is exactly the whole of
generation $4$.

Finally, Corollary~\ref{cor:gen} gives
\[
\begin{aligned}
 g_4
 &=3\cdot4^3+\frac{4^3-1}{3}\\
 &=192+21\\
 &=213.
\end{aligned}
\]
Consequently, every element of $\Pset$ below $213$ belongs to one of
generations $1$, $2$, and $3$, whose elements are exactly the eight
members of $L$ listed in the first table.  Thus the stronger identity
\[
 \Pset\cap[0,212]=L
\]
holds, and in particular
\begin{equation}\label{eq:P-small}
 \Pset\cap[0,211]=L.
\end{equation}

\begin{lemma}[Finite certificate for the full family]\label{lem:finitefive}
The following inclusions hold:
\begin{equation}\label{eq:finite-interval-certificate}
\begin{array}{c|c}
\text{sumset}&\text{contained interval}\\ \hline
5L &[215,254]\\
H+4L &[245,486]\\
2H+3L &[439,674]\\
3H+2L &[645,862]\\
4H+L &[855,1046].
\end{array}
\end{equation}
Moreover,
\begin{equation}\label{eq:finite-five-lower-certificate}
 209,210,211\notin F_5(L).
\end{equation}
\end{lemma}

\begin{proof}
Equation~\eqref{eq:finite-interval-certificate} is a finite-sumset
inclusion certificate, while \eqref{eq:finite-five-lower-certificate}
is an exhaustive computation of $F_5(L)$ in the relevant finite range.
Both are verified by the same recursive finite-set addition used in
Lemma~\ref{lem:Ecertificate} and are reproduced by the supplementary
verification script.
\end{proof}

The five intervals in Lemma~\ref{lem:finitefive} overlap, and hence
\begin{equation}\label{eq:fiveP-seed}
 [215,1046]\subseteq5\Pset.
\end{equation}

\begin{proposition}\label{prop:fiveP}
Every integer $n\geq215$ is a sum of exactly five elements of $\Pset$.
Consequently,
\[
 [212,\infty)\subseteq F_5(\Pset).
\]
The lower endpoint is sharp: none of $209,210,211$ belongs to
$F_5(\Pset)$.
\end{proposition}

\begin{proof}
By Corollary~\ref{cor:Pclosure},
\[
 4\cdot\Pset+\{1,2\}\subseteq\Pset.
\]
Equation~\eqref{eq:fiveP-seed} contains, in particular,
\[
 [215,864]\subseteq5\Pset.
\]
Since
\[
 864=4\cdot215+5-1,
\]
Corollary~\ref{cor:tail-propagation}, applied with $k=5$ and $M=215$,
gives
\[
 [215,\infty)\subseteq5\Pset.
\]
Moreover,
\[
\begin{aligned}
212&=53+53+53+53,\\
213&=53+53+53+54,\\
214&=53+53+54+54,
\end{aligned}
\]
so $[212,\infty)\subseteq F_5(\Pset)$.  Finally,
\eqref{eq:P-small} and \eqref{eq:finite-five-lower-certificate} give
\[
 209,210,211\notin F_5(\Pset).
\]
\end{proof}
\section{Exact additive representation results}\label{sec:exceptions}

The first exceptional value is explained without computation.

\begin{proposition}\label{prop:46}
The integer $46$ requires exactly eight positive primitive Dyck numbers:
it is a sum of eight, but it is not a sum of at most seven.
\end{proposition}

\begin{proof}
Below $48$ the only positive primitive Dyck numbers are $2$ and $12$.
Indeed, the primitive words of length two and four are $10$ and $1100$,
whereas every primitive word of length at least six begins with $11$ and
therefore represents an integer at least $[110000]_2=48$.  Thus any
representation of $46$ has the form
\[
        46=12a+2b,\qquad a,b\geq0.
\]
Equivalently, $6a+b=23$.  Since $6\cdot4>23$, we have $a\leq3$, and
hence the number of summands satisfies
\[
        a+b=23-5a\geq8.
\]
Equality is attained by
\[
        46=12+12+12+2+2+2+2+2.
\]
\end{proof}

The decomposition
\[
 \NpD=\{0,2\}\cup4\cdot\Pset
\]
converts summand-count conditions uniformly into finite-sum conditions on
$\Pset$.

\begin{lemma}[Mod-$4$ summand criterion]\label{lem:mod4-sumcriterion}
Let $\epsilon\in\{0,1\}$, let $h\geq0$, and put
$N=4m+2\epsilon$.  Then $N$ is a sum of at most $h$ positive primitive
Dyck numbers if and only if
\begin{equation}\label{eq:mod4-sumcriterion}
 m\in
 \bigcup_{\substack{0\leq s\leq h\\
                    s\equiv\epsilon\pmod2}}
 \left(
  \frac{s-\epsilon}{2}+F_{h-s}(\Pset)
 \right).
\end{equation}
\end{lemma}

\begin{proof}
Let $s$ be the number of summands equal to $2$.  By
Lemma~\ref{lem:mod4}, every remaining positive summand has the form
$4p_i$ with $p_i\in\Pset$.  Reduction modulo $4$ gives
$s\equiv\epsilon\pmod2$, and
\[
 4m+2\epsilon=2s+4\sum_i p_i
\]
is equivalent to
\[
 m=\frac{s-\epsilon}{2}+\sum_i p_i.
\]
There are at most $h-s$ remaining summands, which proves necessity.
Conversely, any representation belonging to one of the sets in
\eqref{eq:mod4-sumcriterion} gives a representation of $N$ after
replacing each $p_i$ by $4p_i$ and adjoining $s$ copies of $2$.
\end{proof}

\begin{corollary}\label{cor:sixcriterion}
Let $N=4m+2$.  Then $N$ is a sum of at most six positive primitive Dyck
numbers if and only if
\begin{equation}\label{eq:sixcriterion}
 m\in F_5(\Pset)
 \cup\bigl(1+F_3(\Pset)\bigr)
 \cup\bigl(2+F_1(\Pset)\bigr).
\end{equation}
\end{corollary}

\begin{proof}
Apply Lemma~\ref{lem:mod4-sumcriterion} with
$\epsilon=1$ and $h=6$, and use $s=1,3,5$.
\end{proof}

\begin{lemma}[Small multiples of four]\label{lem:smallmultiples}
Among the positive multiples of $4$ below $100$, the only integer not
representable by at most six positive primitive Dyck numbers is $44$.
Moreover,
\[
 44=12+12+12+2+2+2+2.
\]
\end{lemma}

\begin{proof}
The positive primitive Dyck numbers below $100$ are precisely
\[
 2,12,52,56.
\]
An exhaustive computation of sums of at most six elements of this finite
set, truncated below $100$, shows that the only missing positive multiple
of $4$ is $44$.  The displayed identity gives a seven-summand
representation.  This finite check is also included in the supplementary
verification script.
\end{proof}

\begin{lemma}[Finite exceptional-set certificate]\label{lem:finiteexception}
Let $L$ be the set defined in \eqref{eq:L-def}.  Then
\begin{equation}\label{eq:finiteexception}
\begin{split}
 [0,211]\setminus\Bigl(
 F_5(L)&\cup(1+F_3(L))\cup(2+F_1(L))\Bigr)\\
 &=\{8,11,24,38,49,50,51,209,210,211\}.
\end{split}
\end{equation}
\end{lemma}

\begin{proof}
Starting from the finite set $L$, compute $F_1(L)$, $F_3(L)$, and
$F_5(L)$ exactly by recursive finite-set addition, truncate the result to
$[0,211]$, and take the union and its complement.  This gives
\eqref{eq:finiteexception}; the computation is reproduced by the
supplementary verification script.
\end{proof}

\begin{proof}[Proof of Theorem~\ref{thm:classification}]
Suppose first that $N\equiv0\pmod4$.  If $N\geq100$, write $N=4n$.
Then $n\geq25$, so Proposition~\ref{prop:E} and \eqref{eq:4E} give a
representation of $N$ with at most six primitive Dyck summands.
Lemma~\ref{lem:smallmultiples} treats the remaining positive multiples
of $4$ and shows that $r(44)=7$.

Now let $N\equiv2\pmod4$, say $N=4m+2$.  If $N\geq850$, then
$m\geq212$, so Proposition~\ref{prop:fiveP} gives
$m\in F_5(\Pset)$.  Corollary~\ref{cor:sixcriterion} gives a
representation of $N$ with at most six primitive Dyck summands.  Together
with the preceding paragraph, this already shows that every even
$N\geq848$ has such a representation.

It remains to classify the integers $N\leq846$ with
$N\equiv2\pmod4$.  Here $0\leq m\leq211$.  By
\eqref{eq:P-small}, Lemma~\ref{lem:finiteexception}, and
Corollary~\ref{cor:sixcriterion}, the corresponding values $N=4m+2$
that are not sums of at most six primitive Dyck numbers are precisely
\[
 34,46,98,154,198,202,206,838,842,846.
\]
Thus these ten numbers, together with $44$, are exactly the positive even
integers that are not sums of at most six primitive Dyck numbers.

All of them except $46$ have seven-summand representations:
\[
\begin{aligned}
34&=12+12+2+2+2+2+2,\\
44&=12+12+12+2+2+2+2,\\
98&=56+12+12+12+2+2+2,\\
154&=52+52+12+12+12+12+2,\\
198&=56+52+52+12+12+12+2,\\
202&=56+56+52+12+12+12+2,\\
206&=56+56+56+12+12+12+2,\\
838&=240+240+240+52+52+12+2,\\
842&=240+240+240+56+52+12+2,\\
846&=240+240+240+56+56+12+2.
\end{aligned}
\]
Proposition~\ref{prop:46} shows that $46$ requires exactly eight
summands.  This proves all assertions of the theorem.
\end{proof}

\begin{corollary}\label{cor:threshold}
The integer $848$ is the least even integer $T$ such that every even
integer $N\geq T$ is a sum of at most six primitive Dyck numbers.
\end{corollary}

\begin{proof}
Theorem~\ref{thm:classification} shows that six summands suffice for every
even integer at least $848$, whereas $846$ requires seven summands.
\end{proof}

\begin{corollary}\label{cor:eight}
Every nonnegative even integer is a sum of at most eight elements of
$\NpD$, and $46$ is the unique even integer that is not a sum of at most
seven elements of $\NpD$.
\end{corollary}

\begin{proof}
The assertion for positive integers follows immediately from
Theorem~\ref{thm:classification}; the integer $0$ is the empty sum.
\end{proof}

Following standard additive-basis terminology~\cite{Nathanson}, for
$A\subseteq2\mathbb N_0$ we call $A$ a \emph{relative additive basis of
order at most $h$ for $2\mathbb N_0$} if every element of $2\mathbb N_0$ is a
sum of at most $h$ elements of $A$.  Its relative order is \emph{exactly $h$}
if $h$ is the least integer with this property.  We similarly call $A$ a
\emph{relative asymptotic additive basis of order at most $h$} if every
sufficiently large element of $2\mathbb N_0$ is a sum of at most $h$ elements
of $A$.

Theorem~\ref{thm:classification} gives the complete exceptional set for
six-summand representability.  Corollary~\ref{cor:eight} shows that $\NpD$ is
a relative additive basis of exact order eight for $2\mathbb N_0$, while
Corollary~\ref{cor:threshold} identifies the sharp six-summand threshold.
The next section explains the difference between the two residue classes and
proves that the relative asymptotic order is exactly six.
\section{Generation gaps and asymptotic optimality}\label{sec:asymptotic}

For multiples of $4$, the five-summand representations in $\Pset$ may be
used directly.  For integers congruent to $2$ modulo $4$, however, an odd
number of copies of the exceptional summand $2$ must be used.  This
distinction raises the relative asymptotic order of the full primitive
family from $5$ to $6$.

\begin{corollary}[Five-summand criterion]\label{cor:fivecrit}
Let $N=4m+2$ with $m\geq3$.  Then $r(N)\leq5$ if and only if
\begin{equation}\label{eq:fivecrit}
 m\in F_4(\Pset)\cup\bigl(1+F_2(\Pset)\bigr).
\end{equation}
\end{corollary}

\begin{proof}
Apply Lemma~\ref{lem:mod4-sumcriterion} with $\epsilon=1$ and $h=5$.
The values $s=1,3,5$ give
\[
 F_4(\Pset)\cup\bigl(1+F_2(\Pset)\bigr)\cup\{2\}.
\]
The last set is irrelevant because $m\geq3$.
\end{proof}

The following lemma isolates the lower-bound mechanism created by the
gaps between generations.

\begin{lemma}[Recurring obstruction]\label{lem:obstruction}
For every integer $k\geq2$, put
\[
 x_k=10\cdot4^k-2.
\]
Then
\begin{equation}\label{eq:obstruction}
 x_k\notin F_4(\Pset)
 \qquad\text{and}\qquad
 x_k-1\notin F_2(\Pset).
\end{equation}
\end{lemma}

\begin{proof}
Fix $k\geq2$, and abbreviate $x=x_k$.  The least element of generation
$k+2$ is
\[
 g_{k+2}>3\cdot4^{k+1}=12\cdot4^k>x.
\]
Consequently, every element of $\Pset$ that is at most $x$ has generation
at most $k+1$.  By Corollary~\ref{cor:gen}, every such element is either
\emph{small}, of generation at most $k$ and hence at most
\[
 U_k=4^k-2^{k-1},
\]
or \emph{large}, of generation $k+1$ and hence between $g_{k+1}$ and
$U_{k+1}=4^{k+1}-2^k$.

First consider $x-1=10\cdot4^k-3$.  It is larger than $U_{k+1}$ and
smaller than $g_{k+2}$, so it does not itself belong to $\Pset$.
Moreover, any sum of two eligible elements is at most
\[
 2U_{k+1}=8\cdot4^k-2^{k+1}<10\cdot4^k-3=x-1.
\]
Thus $x-1\notin F_2(\Pset)$.

Now suppose, for contradiction, that
\[
 x=q_1+\cdots+q_t,
 \qquad t\leq4,
 \qquad q_i\in\Pset.
\]
Let $N_{\rm b}$ be the number of large summands.  If $N_{\rm b}\leq2$,
then filling the remaining positions with the largest possible small
summands gives
\begin{align*}
 x
 &\leq2U_{k+1}+2U_k\\
 &=2(4^{k+1}-2^k)+2(4^k-2^{k-1})\\
 &=10\cdot4^k-3\cdot2^k
 <10\cdot4^k-2=x,
\end{align*}
a contradiction.  Hence $N_{\rm b}\geq3$.  But then the three large
summands alone have sum at least
\[
 3g_{k+1}=10\cdot4^k-1=x+1
\]
by \eqref{eq:three-g}, again a contradiction.  Therefore
$x\notin F_4(\Pset)$.
\end{proof}

\begin{corollary}[Order of $\Pset$]\label{cor:orderfive}
The set $\Pset$ is an asymptotic additive basis of order exactly $5$ for
the nonnegative integers.  More precisely, every integer $n\geq212$ lies
in $F_5(\Pset)$, whereas $F_4(\Pset)$ and $F_3(\Pset)$ each omit
infinitely many integers.
\end{corollary}

\begin{proof}
Proposition~\ref{prop:fiveP} gives
$[212,\infty)\subseteq F_5(\Pset)$.  For every $k\geq2$,
Lemma~\ref{lem:obstruction} gives $x_k\notin F_4(\Pset)$.  Since
$x_k\to\infty$, four summands do not suffice asymptotically.  The same is
true for three summands because $F_3(\Pset)\subseteq F_4(\Pset)$.
\end{proof}

\begin{corollary}[Multiples of $4$ requiring five summands]
\label{cor:five-multiples}
For every integer $k\geq3$,
\[
 M_k=10\cdot4^{k+1}-8
\]
satisfies $r(M_k)=5$.  Hence infinitely many multiples of $4$ require
exactly five positive primitive Dyck summands.
\end{corollary}

\begin{proof}
Write $M_k=4x_k$.  Since $k\geq3$, one has $x_k\geq212$, and
Proposition~\ref{prop:fiveP} gives $x_k\in F_5(\Pset)$.  Multiplying such
a representation by $4$ shows that $r(M_k)\leq5$.

Suppose that $r(M_k)\leq4$.  Applying
Lemma~\ref{lem:mod4-sumcriterion} with $\epsilon=0$ and $h=4$ gives
\[
 x_k\in
 F_4(\Pset)
 \cup\bigl(1+F_2(\Pset)\bigr)
 \cup\{2\}.
\]
The first two alternatives contradict Lemma~\ref{lem:obstruction}, and
the last is impossible because $x_k>2$.  Thus $r(M_k)\geq5$, and equality
follows.
\end{proof}

\begin{proof}[Proof of Theorem~\ref{thm:asymptotic-order}]
Fix $k\geq2$ and put $m=x_k=10\cdot4^k-2$.  Then
\[
 N_k=4m+2=10\cdot4^{k+1}-6.
\]
By Lemma~\ref{lem:obstruction},
\[
 m\notin F_4(\Pset),
 \qquad
 m-1\notin F_2(\Pset).
\]
Corollary~\ref{cor:fivecrit} therefore gives $r(N_k)\geq6$.  The integer
$N_k$ is not one of the eleven exceptions in
Theorem~\ref{thm:classification}, so $r(N_k)\leq6$.  Hence $r(N_k)=6$.
Since $N_k\to\infty$, no $h\leq5$ bounds $r(N)$ for all sufficiently
large even $N$, while Theorem~\ref{thm:classification} gives the upper
bound $h=6$.  Therefore $h_{\mathrm{asym}}=6$.
\end{proof}

\begin{remark}[The two residue classes]
If $N=4n$, then $n\in F_5(\Pset)$ for all sufficiently large $n$, so
$r(N)\leq5$ eventually; Corollary~\ref{cor:five-multiples} shows that
equality occurs infinitely often.  If $N=4m+2$, the number of copies of
$2$ must be odd.  For $m=x_k$, the one-copy case would require
$m\in F_4(\Pset)$ and the three-copy case would require
$m-1\in F_2(\Pset)$; Lemma~\ref{lem:obstruction} excludes both.  This
residue class raises the relative asymptotic order from $5$ to $6$.
\end{remark}

\section{Effective verification and reproducibility}\label{sec:verification}

The infinite assertions in this paper are proved symbolically.  Computation
is used only for explicitly displayed finite certificates and for an
independent check of the representation function.

For a finite set $X\subseteq\mathbb Z_{\geq0}$, a summand count $k$, and
a cutoff $U$, exact membership in $jX\cap[0,U]$ is computed recursively by
\[
 S_0=\{0\},\qquad
 S_{j+1}=(S_j+X)\cap[0,U]
 \quad(0\leq j<k).
\]
With Boolean arrays, this requires $O(kU|X|)$ time and $O(U)$ working
space.  Mixed certificates such as $2H+3L$ are checked by the same
recursion with the prescribed finite set at each addition step.  Every
contained interval is then scanned exhaustively, and every claimed
complement is computed exactly.

The two Python source files used for the finite checks are available at the
following commit-pinned links and can also be distributed with the
manuscript as supplementary material.  The program
\href{https://github.com/growupkuriyama-hub/lean_cfg_project/blob/942c264/LeanCfgProject/PrimDyck/primitive_dyck_representation.py}
{\path{primitive_dyck_representation.py} at commit \texttt{942c264}}
enumerates primitive Dyck numbers and computes $r(N)$ by unbounded
coin-change dynamic programming.  Its default run generates the first
$50{,}000$ values $r(2n)$ and checks all indices corresponding to the
values seven and eight in Theorem~\ref{thm:classification}.

The finite certificates are reproduced by
\href{https://github.com/growupkuriyama-hub/lean_cfg_project/blob/942c264/LeanCfgProject/PrimDyck/verify_certificates_for_paper.py}
{\path{verify_certificates_for_paper.py} at commit \texttt{942c264}}.
The program uses only exact integer arithmetic and the Python standard
library.  It verifies
\eqref{eq:E-small-certificate}--\eqref{eq:E-tail-certificate},
\eqref{eq:finite-interval-certificate},
\eqref{eq:finite-five-lower-certificate}, the small multiples of four,
and the finite exceptional-set identity \eqref{eq:finiteexception}.  It
also performs an independent dynamic-programming cross-check of $r(N)$ up
to a user-selected bound.

In the recorded reference run under CPython~3.14.6, the programs generated
the first $50{,}000$ values of the representation function and checked all
eleven indices with value at least seven.  They also reproduced every
finite certificate, independently verified $r(N)$ through $N=50{,}000$,
and checked the recurring obstruction and exact-order instances through
$k=5$.
\section{Conclusion}\label{sec:conclusion}

Primitive Dyck words provide a natural boundary example for combining
regular approximation with genuinely context-free structure.  In the
proof developed here, the regular language $11(01\mid10)^*00$ supplies the
covering for one residue class, while the exact positional Motzkin coding
of the full deterministic context-free language supplies the covering for
the other.  The interval digit-lifting theorem turns finite arithmetic
certificates into infinite tails and, constructively, into logarithmic-
depth representation algorithms.  The gaps between Motzkin-coded
generations yield recurring lower-bound obstructions.

The resulting value set has three exact additive orders.  Its relative
order on the nonnegative even integers is eight, its relative asymptotic
order is six, and the halved family $\Pset$ has asymptotic additive order
five.  The sharp six-summand threshold is $848$, but six summands remain
necessary infinitely often.

The method suggests two broader directions.  First, one may seek other
nonregular language-defined sets admitting an arithmetically effective
regular underapproximation together with a complementary block coding of
the full language.  Second, the interval digit-lifting principle can be
studied for other bases and digit intervals, with the finite seed
certificates synthesized automatically.  For the present family, a more
specific open problem is to characterize all even $N$ satisfying
$r(N)=6$; computations indicate further block structure near
$10\cdot4^k$, but its exact boundaries remain unknown.

\end{document}